%% file: main.tex
\documentclass[11pt]{article}

\usepackage{fullpage,amsthm}

\usepackage{graphicx} 
\usepackage{array} 
\usepackage{amsmath, amssymb, amsfonts, color, verbatim}
\usepackage{hyphenat,epsfig,subfigure,multirow}
\usepackage{hyperref}

\usepackage{url}
\usepackage{xspace}
\usepackage[mathscr]{euscript}
\usepackage{algorithm2e, algorithmic}
\usepackage{cite}
\usepackage{enumerate}

\newtheorem{theorem}{Theorem}
\newtheorem{lemma}{Lemma}[section]
\newtheorem{proposition}[lemma]{Proposition}

\newtheorem{claim}[lemma]{Claim}

\newtheorem{definition}{Definition}
\newtheorem{remark}[lemma]{Remark}
\newtheorem*{claim*}{Claim}

\renewcommand{\qed}{\nobreak \ifvmode \relax \else
      \ifdim\lastskip<1.5em \hskip-\lastskip
      \hskip1.5em plus0em minus0.5em \fi \nobreak
      \vrule height0.75em width0.5em depth0.25em\fi}

\newcommand{\etal}{{\it et al.\,}}

\newcommand{\eps}{\epsilon}

\def\poly{\mathop{\rm{poly}}\nolimits}

\newcommand{\Ex}{\mathbf{E}}
\newcommand{\FG}{\ensuremath{\mathcal{G}}}
\newcommand{\card}[1]{\left\vert{#1}\right\vert}
\newcommand{\event}{\mathcal{E}}

\newcommand{\rs}{Ruzsa-Szemer\'{e}di\xspace}

\usepackage[T1]{fontenc}
\usepackage[utf8]{inputenc}

\newcommand{\ourinfo}[1]{Department of Computer and Information Science, University of Pennsylvania. Email: \texttt{#1}.}
\input{macros}

\title{Tight Bounds for Linear Sketches of Approximate Matchings} 
\author{Sepehr Assadi\thanks{\ourinfo{sassadi@cis.upenn.edu}}\and 
Sanjeev Khanna\thanks{\ourinfo{sanjeev@cis.upenn.edu}} \and 
Yang Li\thanks{\ourinfo{yangli2@cis.upenn.edu}} \and 
Grigory Yaroslavtsev\thanks{Warren Center for Network and Data Sciences, University of Pennsylvania. Email: \texttt{grigory@grigory.us}. This work was done while the author was supported by a postdoctoral fellowship at the Warren Center for Network and Data Science.}
}

\date{}

\begin{document}

\maketitle
\thispagestyle{empty}
\input{abstract}

\clearpage
\setcounter{page}{1}
\input{intro}

\input{prelim}

\input{upper}

\input{lower}

\input{conclusion}

\bibliographystyle{acm}
\bibliography{general}

\clearpage
\appendix
\input{appendix}

\end{document}

%% file: macros.tex
\renewenvironment{proof}{\paragraph{Proof.}}{\hfill$\qed$}

\newcommand{\toShrink}{-.25cm} 
\newcommand{\toShrinkEnu}{-.35cm} 
\newcommand{\toShrinkEqn}{-.45cm}

\newcommand{\IR}{\ensuremath{\mathbb{R}}}

\newcommand{\opt}{\ensuremath{\mbox{\tt opt}}}

\newcommand{\est}[1]{\tilde{#1}}

\newcommand{\ceil}[1]{{\left\lceil{#1}\right\rceil}}
\newcommand{\floor}[1]{{\left\lfloor{#1}\right\rfloor}}

\newcommand{\set}[1]{\ensuremath{\{ #1 \}}}

\renewcommand{\card}[1]{{\mid #1 \mid}}

\newcommand{\polylog}{\mbox{\rm  polylog}}

\newcommand{\REM}[1]{}
\newcommand{\Ot}{\tilde{O}}
\newcommand{\Omgt}{\tilde{\Omega}}

\renewcommand{\card}[1]{\left\vert{#1}\right\vert}

\newcolumntype{?}{!{\vrule width 1pt}}

\newcommand{\textbox}[2]{
{
\begin{center}
\begin{tabular}{?p{\linewidth}?}
\hline 
\vspace{0.1cm}
\textbf{#1}
\\
\vspace{0.2cm}
\vspace{-0.2cm}
#2
\vspace{-0.5cm}
\\\hline
\end{tabular}
\end{center}
}
}

\newcommand{\eat}[1]{}

\newcommand{\lzSampler}{$\ell_0$-sampler\xspace}
\newcommand{\lzSamplers}{$\ell_0$-samplers\xspace}

\renewcommand{\Pr}{\mbox{\textnormal{Pr}}}
\renewcommand{\Ex}{{\mbox{\textnormal{E}}}}

\newcommand{\estopt}{\ensuremath{\est{\opt}}}

\newcommand{\algUpper}{1\xspace}

\newcommand{\PS}[1]{\ensuremath{P^{(#1)}}}

\renewcommand{\FG}{\mathcal{G}}

\newcommand{\FR}{\mathcal{R}}
\newcommand{\FL}{\mathcal{L}}
\newcommand{\FM}{\mathcal{M}}

%% file: abstract.tex
\abstract{We resolve the space complexity of linear sketches for approximating the maximum matching
problem in dynamic graph streams where the stream may include both edge insertion and deletion.
Specifically, we show that for any $\eps > 0$, there exists a one-pass streaming algorithm, which
only maintains a linear sketch of size $\Ot(n^{2-3\epsilon})$ bits and recovers an
$n^\epsilon$-approximate maximum matching in dynamic graph streams, where $n$ is the number of
vertices in the graph.  In contrast to the extensively studied insertion-only model, to the best of
our knowledge, no non-trivial single-pass streaming algorithms were previously known for approximating
the maximum matching problem on general dynamic graph streams.
	
Furthermore, we show that our upper bound is essentially tight. Namely, any linear sketch for
approximating the maximum matching to within a factor of $O(n^\eps)$ has to be of size
$n^{2-3\epsilon -o(1)}$ bits. We establish this lower bound by analyzing the corresponding
simultaneous number-in-hand communication model, with a combinatorial construction based on \rs
graphs.  }

%% file: intro.tex
\section{Introduction}\label{sec:intro}
Massive datasets routinely arise in various application domains such as web-scale graphs and social
networks. The space requirement for performing computations on these massive datasets can easily
become prohibitively large. A common way of managing the space requirement is to consider 
algorithms in the \emph{streaming model} of computation.  In this model, formally introduced in the
seminal work of~\cite{AlonMS96}, an algorithm is allowed to make a single or a few passes over the
input while using space much smaller than the input size. We refer the reader to~\cite{Muth05} for a
survey of classical results in this model.

In recent years, there has been extensive work on design of streaming algorithms for various
graph problems, including connectivity, minimum spanning trees, spanners, sparsifiers, matchings,
etc. (see the survey by McGregor~\cite{M14} for a summary of these results).  Two types of graph
streams are mainly studied in the literature: in the \emph{insertion-only model}, the stream
contains only edge insertion, and in the \emph{dynamic model}, the stream contains both edge
insertion and deletion. The focus of this paper is the dynamic model.  The input in this model,
called \emph{dynamic graph streams}, can be defined formally as follows.

\begin{definition}[\!\!\cite{AGM12}]
  A \emph{dynamic graph stream} $S = \langle a_1, a_2, \dots ,a_t\rangle$ defines a multi-graph
  $G(V,E)$ on $n$ vertices $V=[n]$. Each $a_k$ is a triple $(i_k,j_k,\Delta_k)$ where $i_k,j_k \in
  [n]$ and $\Delta_k \in \set{-1,+1}$.  The multiplicity of an edge $(i,j)$ is defined to be:
  \[
  A(i,j) = \sum_{a_k: i_k = i~\wedge~j_k = j} \Delta_k
  \]
  The multiplicity of every edge is required to be always non-negative.
\end{definition}

The streaming model where the frequency of every entry is always non-negative is standard for graph
problems, and this model is generally referred to as the \textit{strict turnstile} model in the
literature (as opposed to the \emph{turnstile} model, which allows negative frequencies also). In this
paper, we study the \emph{maximum matching} problem for \emph{dynamic graph streams} in which the
algorithm is only allowed to make a single pass over the stream.

Matchings have received a lot of attention in the graph stream
literature~\cite{M05,FKMSZ05,EKS09,ELMS11,GoelKK12,KMM12,Z12,AGM12,AG13,K13,GO13,KKS14,CS14,EHLMO15}.
We briefly summarize the previous results for \emph{adversarially ordered} streams. A weaker notion
of \textit{randomly ordered} streams (which is less relevant to our work) is also often considered;
for results in this model, we refer the reader to~\cite{KMM12,KKS14} and references therein.

For the problem of recovering a maximum matching in bipartite graphs, a trivial lower bound on the
space complexity of any streaming algorithm is $\Omega(n)$, which is required for just storing the
matching edges.  Therefore, this problem is usually studied in the \emph{semi-streaming model}
(originally introduced by Feigenbaum~\etal~\cite{FKMSZ05}), where the algorithm is allowed to use
$O(n \cdot \polylog{(n)})$ bits of space. Moreover, no exact algorithm that uses $o(n^2)$ space can
exist~\cite{FKMSZ05}. This motivates the study of $\alpha$-approximate algorithms that output a
matching of size within a multiplicative factor $\alpha$ of the optimum.  For single-pass
semi-streaming algorithms in the insertion-only model, the best known approximation factor is $2$,
which is obtained by simply maintaining a maximal matching during the stream.  On the negative side,
it is shown by~\cite{GoelKK12,K13} that any streaming algorithm that achieves an approximation
factor of better than $e/(e - 1)$ requires the storage of $n^{1+\Omega(1/\log{\log{n}})}$ bits.  For
dynamic graph streams, to the best of our knowledge, no non-trivial single-pass streaming algorithm
using space $o(n^2)$ was known. Resolving the space complexity of matchings in single-pass dynamic
graph streams has been posed as an open problem at the Bertinoro workshop on sublinear and streaming
algorithms in 2014~\cite{TurnstileMatchingOP}.

For the problem of estimating the size of a maximum matching, a strongly sublinear $o(n)$ space
regime has been considered. In the single-pass insertion-only model, when edges arrive in an
adversarial order, the only known positive result for estimating the matching size is that
of~\cite{EHLMO15} which showed that a constant factor approximation is possible in $\Ot(n^{2/3})$
space under the assumption that the underlying graph is planar. The same paper~\cite{EHLMO15} also
provides a lower bound of $\Omega(\sqrt{n})$ (resp. $\Omega(n)$) bits of space for randomized
(resp. deterministic) algorithms that approximate the matching size in bipartite graphs to within a
factor of $3/2$. For the state of the art in the streaming model which allows multiple passes over
the stream, we refer the reader to~\cite{AGM12,AG13,GO13,AhnG13,K13} and references therein.

To the best of our knowledge, the only result concerning matchings in the single-pass dynamic graph
streams is the recent paper by Chitnis~\etal~\cite{CCHM15}, which provides an algorithm for
computing a maximal matching of size $k$ using $\Ot(nk)$ space. For multi-pass dynamic graph
streams, \cite{AhnG13} provides a $(1-\eps)$-approximation scheme for the weighted non-bipartite
matching problem using $O(p/\eps)$ passes with $O(n^{1+1/p})$ space (see also~\cite{M14}).

Finally, closely related to our work is a recent line of work on communication complexity of
approximate matchings in the multi-party setting~\cite{DNO14,ANRW15,HRVZ15}.  The one that is
closest to ours is~\cite{HRVZ15}, which shows a tight bound of
$\Theta\left(\frac{nk}{\alpha^2}\right)$ on the total communication required to compute an
$\alpha$-approximate matching for bipartite graphs, in the $k$-party message passing model where the
edges of the input graph are arbitrarily partitioned between the players.

\paragraph{Linear sketches.}
One of the most powerful techniques for designing streaming algorithms is \textit{linear sketching}.
Let $n$ be the number of vertices in the input graph.  Then edge multiplicities can be treated as a
vector $f \in \mathbb R^{\binom{n}{2}}$ with entries $f_e$.  Let $A \in \mathbb R^{d \times
  \binom{n}{2}}$ be a (possibly randomly chosen) matrix.  Then $A\cdot f$ is referred to as a
\textit{linear sketch} of the input stream. If all that a streaming algorithm maintains is such a
linear sketch, then the space requirement of the algorithm is proportional to $d$.  On any incoming
update $(i_k,j_k,\Delta_k)$, the linear sketch will be updated to $A\cdot f' = A\cdot f + \Delta_k
\cdot A \cdot \mathbf{1}_{(i_k,j_k)}$ where $f'$ is the new vector of edge multiplicities and
$\mathbf{1}_{(i_k,j_k)} \in \IR^{\binom{n}{2}}$ is a unit vector whose only non-zero entry is the
$(i_k,j_k)$ entry.  At the end of the stream, the algorithm can apply an \emph{arbitrary function}
to the linear sketch to compute the final answer.

Linear sketching is the only existing technique for designing streaming algorithms in the turnstile
model and even for dynamic graph streams\footnote{To the best of our knowledge the only exception is
  the recent paper~\cite{CCHM15}, which considers a \emph{promised problem} in dynamic graph
  streams.  However, it is worth mentioning that for the non-promise version of the problem, the
  algorithm given in the same work can again be viewed as a linear sketching algorithm.}.  Linear
sketches are also one of the main techniques for designing mergeable summaries~\cite{ACHPWY13} used
in distributed computing.  These facts have made linear sketches a computational model of their own.
Multiple results are known about the power and limitations of linear sketches,
e.g.~\cite{AhnGM12Linear,AGM12,ANPW13,HW13,KLMMS14}.  In fact, it is shown that any \emph{one-pass
  turnstile streaming algorithm} can be implemented by maintaining only a linear sketch of the input
during the stream~\cite{LNW14}\footnote{We emphasize that the result in~\cite{LNW14} is proven for
  the turnstile model rather than the \emph{strict} turnstile model.}. For an in-depth introduction
of linear sketching and its applications for dynamic streams and distributed computing, we refer the
reader to recent surveys by McGregor~\cite{M14} (graph streams) and Woodruff~\cite{W14}
(computational linear algebra).

\subsection{Our results} \input{results}

%% file: results.tex
We resolve the space complexity of linear sketches for approximating maximum matchings by proving
tight upper and lower bounds on the space requirement.  For the upper bound, we establish the
following theorem.

\begin{theorem}\label{thm:upper-informal}
  There is a single-pass randomized streaming algorithm that takes as input a
  parameter $0 < \eps <1/2$ and a bipartite graph $G$ with $n$
  vertices, specified by a dynamic graph stream, uses $\Ot(n^{2-3\eps})$ bits of space, and
  outputs a matching of size $\Omega(\opt / n^\epsilon)$ with high probability, where $\opt$ is the size
  of a maximum matching in $G$. Moreover, the algorithm only maintains
  a linear sketch during the stream.
\end{theorem}

We prove this result by designing a \emph{sampling based} algorithm that takes advantage of the
well-known linear sketching implementation of \lzSampler (see Section~\ref{sec:lzSampler}).  The
algorithm maintains a set of (edge) samplers that are coordinated in such a way that the sampled
edges are ``well-spread'' across different parts of the graph and hence contain a relatively large
matching. The main challenge is to achieve such a coordination for linear sketching based
samplers. Such a coordination is typically achieved via sequential operations that depend on the
state of the stream, while linear sketches are inherently oblivious to the underlying state.

Note that our algorithm, though stated for bipartite graphs, also works for general graphs by
applying the standard technique of choosing a random bipartition of the vertices upfront and only
considering edges that cross the bipartition, while losing a factor of $2$ in the approximation
ratio.  We further note that for weighted graphs with $\poly(n)$-bounded weights, the standard
``grouping by weight'' technique can be used to obtain a similar result for computing an
approximation to weighted matching, while losing a factor of $O(\log{n})$ in the approximation
ratio.

We complement our upper bound by the following (essentially) matching lower bound.

\begin{theorem}\label{thm:lower-informal}
  There exists a constant $ c > 0$, such that for any $\eps > 0$, any randomized
  \emph{linear sketch} that can be used to recover a matching of size $\opt/(c \cdot n^{\eps})$
  for every input bipartite graph $G$ on $n$ vertices with constant probability, must have
  worst case space complexity of $n^{2-3\eps-o(1)}$ bits. Here, $\opt$ denotes the size of a maximum matching in $G$.
\end{theorem}

This result is obtained as a corollary of our lower bound on the communication complexity of
approximating maximum matchings in the \emph{number-in-hand simultaneous model}
(Theorem~\ref{thm:lower}); see Section~\ref{sec:distributed-computation} for the exact definition of
this model and the connection with linear sketches.

Our construction follows the line of work by~\cite{GoelKK12,K13} on using \rs graphs for proving
lower bound on space complexity of streaming algorithms for maximum matching problem.  However,
focusing on the number-in-hand simultaneous model allows us to benefit from different construction
of \rs graphs that are dense, hence bypassing the limitation of the aforementioned works on proving
lower bound for \emph{larger} approximation ratios and the $n^{1+\Omega(1/\log{\log{n}})}$ barrier
on the value of the space lower bound. We elaborate more on this in Section~\ref{sec:rs}.

Finally, we note that Theorem~\ref{thm:upper-informal} and Theorem~\ref{thm:lower-informal} provide
(essentially) tight bounds on the space complexity of any streaming algorithm for dynamic graph
streams that only maintains a linear sketch during the stream.  This makes progress on an open
problem posed at the Bertinoro workshop on sublinear and streaming algorithms in
2014~\cite{TurnstileMatchingOP}, regarding to the possibility of having constant factor
approximation to the maximum matching in $o(n^2)$ space.

\paragraph{Recent related work.} Independently and concurrently to our work, Konrad~\cite{Konrad15}
has also studied the problem of designing linear sketches for approximating matchings in dynamic
graph streams. Konrad's work shows that an $n^{\eps}$-approximation can be obtained using a linear
sketch of size $\Ot(n^{2-2\eps})$, and it establishes a lower bound of $\Omega(n^{3/2 - 4\eps})$ on
the size of any linear sketch that yields an $n^{\eps}$-approximation. Our approaches for
establishing the lower bound on the sketch size are in the same spirit, though the techniques and
constructions are quite different.

\subsection{Organization} 
In Section~\ref{sec:prelim}, we introduce the key concepts and tools used in this paper. In
particular, Section~\ref{sec:lzSampler} describes \lzSamplers and how we use them in our algorithm;
Section~\ref{sec:distributed-computation} formally defines the number-in-hand simultaneous model and
how it is connected to linear sketches; and Section~\ref{sec:rs} provides a definition of \rs graphs
and the specific construction used in our lower bound construction.  In Section~\ref{sec:upper}, we
describe a single-pass streaming algorithm for the maximum matching problem in dynamic graph streams
and prove Theorem~\ref{thm:upper-informal}. In Section~\ref{sec:lower}, we present our lower bound
construction and Theorem~\ref{thm:lower-informal}. Finally, we conclude our results in
Section~\ref{sec:conc}.

%% file: prelim.tex
\section{Preliminaries}\label{sec:prelim}
\subsection{$\ell_0$-Samplers} \label{sec:lzSampler}
We use the following tool developed in the streaming literature.
\begin{definition}[$\ell_0$-sampler~\cite{FrahlingIS2008}]
  Let $ 0 < \delta < 1$ be a parameter.  An $\ell_0$-sampler is an algorithm which given access to a
  dynamic stream, returns FAIL with probability at most $\delta$, and otherwise, outputs an element
  $e$, along with the frequency $f_e$, where $e$ is uniformly distributed among the \emph{non-zero}
  entries of the frequency vector $f$.
\end{definition}

We use \lzSamplers as follows: For the input graph $G(V,E)$, let $V' \subseteq V$ be a subset of
vertices; suppose we maintain an \lzSampler over the stream where only the edges between vertices in $V'$ are
considered.  At the end of the stream, we can use the \lzSampler to recover one edge between the
vertices in $V'$, if such an edge exists.

We use the following lemma in our algorithm which implements \lzSamplers using linear sketches.

\begin{lemma}[\!\!\cite{JowhariST2011}]\label{lem:lz-sampler}
  For any $ 0 < \delta < 1$, there is a linear sketching implementation of $\ell_0$-sampler for the
  frequency vector $f \in \mathbb R^n$ with probability of success $1-\delta$, using $O(\log^{2}{n}
  \cdot \log{(\delta^{-1})})$ bits of space.
\end{lemma}

\subsection{The Number-in-Hand Simultaneous Model}\label{sec:distributed-computation}
The number-in-hand simultaneous model is defined as follows. The input vector $x = x_1 + \ldots +
x_k$ is partitioned \emph{adversarially} between $k$ different players $\PS{1},\ldots,\PS{k}$, where
each player $\PS{i}$ only sees the input $x_i$.  All players have access to an infinite shared
string of random bits, referred to as \emph{public coins}.  The goal for the players is to compute a
function $f(x)$ by \emph{simultaneously} sending a (possibly randomized using only public
randomness) message to a special party called the \emph{coordinator}, according to a pre-specified
\emph{protocol}. For any input $x$, the coordinator is then required to output $f(x)$ with
probability $1 - \delta$ over the randomness used in the protocol.  We refer the reader
to~\cite{CCbook} for more information about communication complexity in general.

To prove our lower bound in Theorem~\ref{thm:lower-informal}, we consider the maximum matching
problem in the number-in-hand simultaneous model, defined formally as follows. Each player $\PS{i}$
is given a vector $x_i \in \set{0,1}^{\binom{n}{2}}$, representing the edges of a graph
$G_i(V,E_i)$, with $V = [n]$. Their goal is to approximate the maximum matching in the multi-graph
$G(V,E)$, where $E$ is represented by the vector $x = x_1 + \ldots + x_k$.

We should note that space lower bounds for single-pass streaming algorithms are usually obtained by
proving communication complexity lower bounds in a different model of communication, i.e., the
\emph{one-way} communication model, in which player $\PS{1}$ speaks to $\PS{2}$, who speaks to
$\PS{3}$, etc., and finally $\PS{k}$ outputs the answer.  In this model, the maximum matching
problem has a simple $2$-approximation algorithm using $O(n)$ communication per player: send a
maximal matching from each player to the next one. Since we are looking for space complexity of
$n^{2-3\eps - o(1)}$, the one-way model cannot lead to our lower bound in
Theorem~\ref{thm:lower-informal}.

The following proposition enables us to consider the simultaneous model instead of one-way model in
proof of our space lower bound.  This reduction is well-known in the literature (see~\cite{LNW14},
for example).
\begin{proposition}\label{prop:reduction}
  Suppose there is a \emph{linear sketch} of size $s$ bits for a function $f$ from which $f$ can be
  computed with failure probability at most $\delta$; then for any $k \geq 1$, there exists a
  \emph{public-coin number-in-hand simultaneous} protocol for $k$ players to compute $f$, where each
  player communicates a message of size $s$ and the coordinator is able to compute $f$ with failure
  probability at most $\delta$.
\end{proposition}
\begin{proof}
  The players use the public coins to construct the set of random coin tosses required to create the
  matrix $A$ in the linear sketch. Then, each player computes $A\cdot x_i$ and sends it to the
  coordinator. The coordinator can now compute $A\cdot x$ for $x = x_1 + \ldots + x_k$ by simply
  computing $A\cdot x = A\cdot x_1 + \ldots + A\cdot x_k$, and then compute $f(x)$ from $A\cdot x$.
\end{proof}

\subsection{\rs graphs}\label{sec:rs}

Given an undirected graph $G(V,E)$ and a set of edges $F \subseteq E$, we denote by $V(F)$, the set
of vertices which are incident on at least one edge in $F$. Moreover, we denote by $E(F)$, the set
of edges \emph{induced} by $F$, i.e. $E \cap (V(F) \times V(F))$. $F$ is said to be an \emph{induced
  matching} if no two edges in $F$ share an endpoint and $E(F) = F$.

\begin{definition}[\rs graph]\label{def:rs-graph}
  We call a graph $G$ an $(r,t)$-\emph{\rs graph}, $(r,t)$-RS graph for short, if the set of edges
  in $G$ consists of $t$ pairwise disjoint induced matchings $M_1,\ldots,M_t$, each of size $r$.
\end{definition}

In general, graphs of this type are of interest when $r$ and $t$ are relatively large as a function
of number of vertices in the graph.  The first construction of an $(r,t)$-RS graph was given by
Ruzsa and Szemer\'{e}di~\cite{RuszaS78} with parameters $r = \frac{n}{e^{O(\sqrt{n})}}$ and $t =
\frac{n}{3}$.  By now, there are several known construction of these graphs with different range of
parameters $r$ and $t$~\cite{AlonMS12,BirkLM93,FischerLNRRS02} (see~\cite{AlonMS12} for more
information).  In particular, Fischer \etal~\cite{FischerLNRRS02} introduced a construction with
parameters $r = (1-o(1))\cdot\frac{n}{3}$ and $t = n^{\Omega(1/\log{\log{n}})}$.  This construction
was further used and improved by~\cite{GoelKK12,K13} to obtain their aforementioned lower bound of
$n^{1+\Omega(1/\log{\log{n}})}$ on space complexity of streaming algorithms for maximum matching
problem in the insertion-only streams.

We use the construction of $(r,t)$-RS graphs given by Alon~\etal~\cite{AlonMS12}, which is
summarized in the following theorem.

\begin{theorem}[\!\!\cite{AlonMS12}]\label{thm:alon-rs}
  For any sufficiently large $N$, there exists an $(r,t)$-RS graph on $N$ vertices with $r =
  N^{1-o(1)}$ and $r\cdot t = {{N \choose 2}} - o(N^2)$.
\end{theorem}

%% file: upper.tex
\section{An $O(n^\eps)$-approximation using $\Ot(n^{2-3\eps})$ space}\label{sec:upper}

In this section, we present our algorithm for computing an approximate
maximum matching in the dynamic graph streams and prove the following
theorem.

\begin{theorem}\label{thm:dynamic-streaming}
  There is a single-pass randomized streaming algorithm that takes as input a parameter $0 < \eps <
  1/2$ and a bipartite graph $G$ with $n$ vertices specified by a dynamic graph stream, uses
  $O(n^{2-3\eps}\cdot\polylog{(n)})$ bits of space, and with high probability, outputs a matching of
  size $\Omega(\opt/n^\eps)$, where $\opt$ is the size of a maximum matching in $G$.
\end{theorem}

In the following, whenever we use \lzSamplers, we always apply
Lemma~\ref{lem:lz-sampler} with parameter $\delta = n^{5}$. Since the
number of \lzSamplers used by our algorithm is bounded by $O(n^2)$,
with high probability, none of them will fail. In the rest of this
section, we always assume this is the case for all \lzSamplers we use,
and we do not explicitly account for the probability of \lzSamplers
failure in our proofs.

For simplicity, we assume that the algorithm is provided with a value $\estopt$ that is a
$2$-approximation of $\opt$ i.e., the size of a maximum matching in $G$. This is without loss of
generality, since we can run our algorithm for $O(\log{n})$ different estimates of $\opt$ in
parallel and output the largest matching among the matchings found for all estimates. In addition,
we can assume $\opt \ge n^\eps$, since otherwise a single edge is an $n^\eps$-approximation of the
maximum matching, which can be obtained by maintaining an \lzSampler over all edges in the graph.

\textbox{Algorithm \algUpper \textnormal{A single-pass dynamic streaming algorithm for the maximum matching problem.}}{
\medskip
$\bullet$ \textbf{Input:} A bipartite graph $G(L,R,E)$ with $n$
vertices on each side, specified by a dynamic graph stream, a
parameter $0 < \eps < 1/2$, and a $2$-approximation to the size of a maximum matching in
$G$ as $\estopt$.\\ 
$\bullet$ \textbf{Output:} A matching $M$ with size
$\Omega(\opt/n^{\eps})$. 
\\ $\bullet$ \textbf{Pre-processing:}
  \begin{enumerate}
  	\item Let $\alpha = \ceil{\frac{\estopt}{n^{\eps}}}$, $\beta  = 6\ceil{\frac{\estopt}{n^{2\eps}}}\cdot\log{n}$, and $\gamma = 4n^{\eps}$. 
	\item Create two collections $\FL$ and $\FR$, each containing $\alpha$ sets (called \emph{groups}). Create
          two \emph{$\gamma$-wise independent} hash functions $h_L: L \mapsto \FL$ and $h_R: R
          \mapsto \FR$. Assign each vertex $u \in L$ (resp. $v \in R$) to the group $h_L(u) \in \FL$
          (resp. $h_R(v) \in \FR$).
        \item For each $L_i \in \FL$, assign $\beta$ groups in $\FR$ to $L_i$ chosen
          \emph{independently and uniformly at random} with replacement. For each $R_j$ assigned to
          $L_i$, we say $R_j$ is an \emph{active partner} of $L_i$ and $(L_i, R_j)$ form an
          \emph{active pair}.
  \end{enumerate}
$\bullet$ \textbf{Streaming updates:}
  \begin{enumerate}[$*$]
  \item For each $L_i \in \FL$ and each of its active partners $R_j \in \FR$, maintain an \lzSampler
    over the edges between the vertices in $L_i$ and $R_j$.
  \end{enumerate}
$\bullet$ \textbf{Post-processing:} 
  \begin{enumerate}[$*$]
  \item Sample one edge from each maintained \lzSampler and compute a maximum matching $M$ over the
    sampled edges.
  \end{enumerate}
}

The space complexity of Algorithm~\algUpper is easy to verify. The algorithm stores two
$\gamma$-wise independent hash functions $h_L$ and $h_R$ to assign vertices to their
groups, which requires $\Ot(\gamma) = \Ot(n^{\eps})$ bits of space~\cite{RAbook}. $\Ot(\alpha \cdot \beta)$ truly
random bits are needed for identifying the active partners of each group in $\FL$, and $O(\alpha
\cdot \beta)$ \lzSamplers are maintained for the active pairs during the stream, where each of them
requires $O(\log^3{n})$ bits of space (Lemma~\ref{lem:lz-sampler}). Hence, the total space complexity
of the algorithm is:
\[
\Ot(n^\eps + \alpha \cdot \beta) = \Ot(n^\eps + \frac{\opt^2}{n^{3\eps}}) = \Ot(n^{2-3\eps})
\]
where the last equality is by choice of $\eps <1/2$. 

We now prove the correctness of the algorithm. Fix a maximum matching $M^*$ in $G$ with size
$\opt$. The following concentration bound ensures that each group in $\FL$ and $\FR$ contains $(1
\pm 0.001)n^{\eps}$ vertices of the maximum matching $M^*$.

\begin{claim}[\!\!\cite{SchmidtSS95}]\label{clm:independence}
	If $X$ is sum of $k$-wise independent random variables taking values in $[0,1]$, and $\mu = \Ex[X]$, then:
	\[
        \Pr(\card{X-\mu} > \eps'\mu) < \exp(-\floor{k/2}) ~~~~\forall \eps' \leq 1, k \leq
        \floor{\eps'^2\mu e^{-1/3}}
	\]
\end{claim}
For simplicity, in the following, we assume every group has exactly $n^{\eps}$ vertices of
$M^*$\footnote{One can simply substitute $c \in [0.999n^{\eps}, 1.001n^{\eps}]$ in following
  equations instead of $n^{\eps}$ and obtain the same result with a slight change in the
  constants.}. For any group $L_i \in \FL$, (resp. $R_j \in \FR$) we refer to the edges in $M^*$
that are incident on $L_i$ (resp. $R_j$) as the \emph{matching edges} of this group. Since $M^*$ is
a matching, the number of matching edges of each group is also $n^{\eps}$.
 
We say a $(L_i,R_j)$ pair is \emph{matchable} by $M^*$ if $L_i$ and $R_j$ share at least one matching edge.
The general idea of the proof is to show that among all $(L_i, R_j)$ active pairs, there is a subset
$\FM \subseteq \FL \times \FR$ of $\Omega(\opt/n^{\eps})$ active pairs with the following two
properties:
\begin{enumerate}[(i)]
\item \label{prop:1} Each pair is matchable by $M^*$.
\item \label{prop:2} No two pairs in $\FM$ share the same endpoints $L_i$ or $R_j$.
\end{enumerate} 

Intuitively, properties~(\ref{prop:1},\ref{prop:2}) together ensures that there exists a ``matching''
between the groups in $\FL$ and $\FR$ of size $\Omega(\opt/n^\eps)$. Since we maintain an \lzSampler
for each active pair in $\FM$, and each matchable active pair contains at least one edge in $G$, the
\lzSamplers for the matchable active pairs will return $\card{\FM}$ edges, which will form a
matching of size $\card{\FM} = \Omega(\opt/n^\eps)$ in graph $G$.
 
To prove the existence of such a set $\FM$, we start by arguing that there are
$\Omega(\opt/n^{\eps})$ groups $L_i$ in $\FL$ such that (essentially) $\Omega(n^\eps)$ matching
edges of $L_i$ are incident on $\Omega(n^\eps)$ \emph{distinct} groups $\FR$. Consequently, when the
algorithm randomly assigns $L_i$ with $\beta = \Omgt(\alpha/n^{\eps})$ groups in $\FR$, since
$\card{\FR} = \alpha$, with high probability, at least one of the $(L_i,R_j)$ active pairs is
matchable by $M^*$. This ensures that we have $\Omega(\opt/n^{\eps})$ $(L_i, R_j)$ matchable active
pairs where all $L_i$'s are distinct. Finally, we show that a constant fraction of these matchable
active pairs also have distinct $R_j$'s, with a constant probability, proving
property~(\ref{prop:2}).

We now provide the formal proof. To continue, we need the following definitions.  We say a group
$L_i \in \FL$ is \emph{spanning} if the matching edges of $L_i$ are incident on at least
$\min\set{n^\eps,\alpha}/3$ different $R_j \in \FR$. We say that $L_i$ \emph{preserves} an edge in
$M^*$ if $L_i$ belongs to at least one matchable active pair. 

\begin{lemma}\label{lem:spanning-preserves}
  With probability at least $1 - 1/n$, every spanning $L_i \in \FL$ preserves an edge in $M^*$. 
\end{lemma}
\begin{proof}
  We argue that if $L_i$ is spanning, then $L_i$ preserves an edge in $M^*$ with probability at
  least $1- 1/n^2$. Then, by applying union bound over all spanning $L_i$, with probability $(1 - 1/n)$,
  every spanning $L_i$ preserves an edge in $M^*$.
	
  For any spanning $L_i$, there are $\min\set{n^\eps, \alpha}/3$ different $R_j$'s such that $M^*$
  contains an edge between $L_i$ and $R_j$, i.e., $(L_i,R_j)$ is matchable by $M^*$. Recall that
  $L_i$ is assigned with $\beta = (6\alpha\log{n})/n^\eps$ groups in $\FR$ uniformly at random.
          
  If $\alpha/3$ different $R_j$'s are matchable with $L_i$ by $M^*$, assigning $2\log{n}$ random
  groups in $\FR$ to $L_i$ suffices to ensure that with probability at least $1 - 1/n^2$, $L_i$
  preserves an edge in $M^*$.
          
  If $n^\eps/3$ different $R_j$'s are matchable with $L_i$ by $M^*$,
  the probability that a spanning $L_i$ does not preserve any edge in
  $M^*$ is at most
  \[
  (1 - {n^\eps \over 3\alpha})^{\beta} \leq \exp(-{n^\eps \over 3\alpha}\cdot {\beta}) =
  \exp(-{n^\eps \over 3\alpha} \cdot \frac{6\alpha\log{n}}{n^{\eps}} ) \leq {1 \over n^2}
  \]
\end{proof}

\begin{lemma}\label{lem:large-spanning}
  With a constant probability, at least $1/4$ of the $L_i$'s  are spanning.
\end{lemma}
\begin{proof}
  We use the following simple balls and bins argument (see Appendix~\ref{app:full-ball} for a
  proof).
\begin{claim}\label{clm:full-ball}
  Suppose we assign $x$ balls to $y$ bins independently and uniformly at random. With probability at
  least $1/2$, the number of non-empty bins is at least $\min\set{x,y}/3$.
\end{claim}

  Fix an $L_i \in \FL$. Consider each $R_j \in \FR$ as a \emph{bin}
  and each matching edge of $L_i$ as a \emph{ball}. An edge $(u,v)$
  ($v \in R$), i.e., a ball, is assigned to the bin $R_j$ iff the
  group assigned to vertex $v$ is $R_j$. The number of balls here is $n^{\eps}$ and since we use a $\gamma$-wise independent hash function ($\gamma > n^\eps$) to 
  assign the balls to the bins, all these $n^{\eps}$ balls are assigned independently.   
  By Claim~\ref{clm:full-ball}, at least $\min\set{n^\eps, \alpha}/3$
  different $R_j$'s have edges in $M^*$ that are incident on $L_i$ and
  $R_j$ (hence $L_i$ is spanning), with probability at least $1/2$. By
  Markov inequality, with a constant probability, at least $\alpha/4$
  $L_i$'s are spanning.
\end{proof}
	
        \begin{lemma}\label{lem:distinct-spanning}
          With a constant probability, $\Omega(\alpha)$ groups in $\FR$ are active partners of
          \emph{distinct} spanning $L_i \in \FR$, such that $L_i$ and $R_j$ are matchable by $M^*$.
        \end{lemma}
        \begin{proof}
	  Suppose each spanning $L_i$, when picking the $R_j$'s, only
          keeps the first $R_j$ where $L_i$ and $R_j$ are matchable by
          $M^*$ (picking more can only increase the size of the final
          matching).  We only need to show that the number of distinct
          $R_j$'s that are kept by $L_i$'s is $\Omega(\alpha)$.
          
         Suppose $n^\eps < \alpha$; the other case when $n^\eps \ge
         \alpha$ is an easy case since each spanning $L_i$ is
         matchable with $1/3$ fraction of the groups in $\FR$.  
         By Lemma~\ref{lem:large-spanning}, there are $\alpha/4$ spanning $L_i$'s 
         with high probability. Therefore, there are $(n^\eps/3) \cdot (\alpha/4) = (n^{\eps}/3) \cdot
         (\opt/(4n^\eps)) = (\opt/12)$ edges in $M^*$ incident on all
         the spanning $L_i$'s; we denote these $\opt/12$ edges of
         $M^*$ by $M'$.  Since each group in $\FR$ has $n^{\eps}$
         matching edge, and $\alpha = \opt/n^{\eps}$, it must be that
         at least $(\alpha/24)$ groups in $\FR$ contain at least
         $(n^\eps/23)$ vertices incident on $M'$; otherwise, the total
         number of edges incident on $M'$ is less than
	  \[
	  \frac{\alpha}{24} \cdot n^\eps + {\frac{23}{24}} \cdot \frac{n^\eps}{23} = \frac{\opt}{24} + \frac{\opt}{24} = {\opt \over 12}
	  \]
	  Let $\FR'$ be the set of all these $(\alpha/24)$ groups in $\FR$. Conditioned on the event
          that $L_i$ preserves an edge in $M^*$, for each of the $R_j$ groups that are matchable
          with $L_i$ by $M^*$ (there are at most $n^\eps$ such groups), the probability that $R_j$
          is kept by $L_i$ is at least $1/n^\eps$. Therefore, for each of these $(\alpha/24)$ $R_j$
          groups, the probability that $R_j$ is assigned to any spanning $L_i$ is at most
	  \[  
	  (1-{1 \over n^\eps})^{n^\eps/23} \le \exp({-{1n^\eps \over 23 n^\eps}}) = e^{-1/23}
	  \]
	  Hence the expected number of groups in $\FR'$ that are not active partner of any spanning
          $L_i$ is at most $\Omega(\alpha)$. By Markov inequality, with a constant probability,
          $\Omega(\alpha)$ different $R_j \in \FR'$ will be kept by some spanning $L_i$.

	  Note that the probability of success can be boosted to any constant by allowing $L_i$ to
          repeatedly pick $\beta$ groups from $\FR$ as active partners for a constant number of
          times.
\end{proof}

\begin{proof}(Theorem~\ref{thm:dynamic-streaming}) By Lemma~\ref{lem:distinct-spanning},
  $\Omega(\alpha)$ groups in $\FR$ will be assigned to $\Omega(\alpha)$ distinct spanning groups in
  $\FL$; moreover, every such pairs are matchable by $M^*$. Since all these pairs are matchable by
  $M^*$, there exists at least one edge between each of these pairs. By picking one edge for each of
  these pairs, using the \lzSampler between these $(L_i,R_j)$ active pairs, we obtain a matching of
  size $\Omega(\alpha)$. Therefore, in the post-processing step, the algorithm can find a matching
  of size $\Omega(\alpha) = \Omega(\opt/n^{\eps})$.
\end{proof}

%% file: lower.tex
\section{An $n^{2-3\eps - o(1)}$ lower bound for $O(n^{\eps})$-approximation}\label{sec:lower}

In this section, we provide our lower bound result for approximating
the maximum matching using linear sketches.  As stated in
Section~\ref{sec:distributed-computation}, we only need to prove the
lower bound for the number-in-hand simultaneous model; the rest
follows from Proposition~\ref{prop:reduction}.

\begin{theorem}\label{thm:lower}
  There exists a constant $c > 0$, such that for any $\eps > 0$, any
  protocol for approximating the maximum
  matching to within a factor of $c \cdot n^{\eps}$ on every graph $G$
  with $n$ vertices, in the number-in-hand simultaneous model with $k
  = n^{\eps + o(1)}$ players, has to communicate $n^{2-3\eps-o(1)}$
  bits from at least one player.
\end{theorem}

Note that though we state Theorem~\ref{thm:lower} for general graphs, the reduction mentioned after
Theorem~\ref{thm:upper-informal} implies the same lower bound for bipartite graphs.

By Yao's principle, it is enough to prove the lower bound on the communication complexity of
deterministic protocols on some fixed distribution on the inputs (known to the players). We provide
the following distribution as a hard input distribution for every deterministic protocol.

\textbox{The hard input distribution \textnormal{(for any $\eps > 0$
    and any sufficiently large integer $N$)}}{
	\begin{itemize}
		\item \textbf{Parameters:} $r,t,k,n,\alpha$:
		\[
			r = N^{1-o(1)}~~~t = \frac{{N \choose 2} - o(N^2)}{r}~~~k = \left(\frac{N^{1+\eps}}{r}\right)^{1/(1-\eps)}~~~~n=k\cdot N~~~~\alpha = n^{\eps} 
		\]
		\item For each player \PS{i} ($i \in [k]$) independently, 
		\begin{enumerate}
			\item Create a set of $N$ vertices $V_i$ and construct an arbitrary $(r,t)$-RS graph over $V_i$. 
			\item Pick $\lambda \in [t]$ uniformly at random and let $V^{*}_i$ be the set of vertices matched in the induced matching $M_\lambda$.
			\item For each of the $t$ induced matchings, drop half of the edges uniformly at random.
		\end{enumerate}
		\item Pick a random permutation $\pi$ of $[n]$. For every player $\PS{i}$, let the \emph{label} of $v_j$ to be $\pi(j)$ for every $v_j \in V_i \setminus V^{*}_i$ and let 
		the label of $u_j$ to be $\pi(N+(i-2) \cdot r + j)$ for $u_j \in V^{*}_i$. Note that the vertices with the same label correspond to the same vertex in the final graph. 
		\end{itemize}
}

Several remarks are in order. First, one can easily verify the
following relation between the parameters,
\[
	k = \alpha N / r = n^{\eps}N^{o(1)} = n^{\eps + o(1)}
\]
Second, for the choice of the parameters $r,t,$ and $N$, by
Theorem~\ref{thm:alon-rs}, such an $(r,t)$-RS graph with $N$ vertices
indeed exists. Moreover, note that the vertices in $V^{*}_i$ for all
players are assigned with unique labels, while the vertices in $V_i
\setminus V^{*}_i$ are assigned with the same set of labels. Consequently, 
the final graph is a multi-graph with $n$ vertices and $O(kN^2) = O(n^{2 - \eps - o(1)})$ total
number of edges (counting the multiplicties). We now
briefly describe the intuition behind this distribution.

Each player $\PS{i}$ is given an $(r,t)$-RS graph with half of the
edges discarded uniformly at random from each of the $t$ induced
matchings.  Moreover, only a single induced matching is ``private''
and the vertices that are not incident on this matching are shared
among all players. In addition, the identities of the private matching
and shared vertices are unknown to the players. Intuitively, for any
deterministic protocol over this distribution, every player has to
send enough information for the coordinator to recover a large
fraction of the edges from every induced matching; otherwise, the
coordinator will not be guaranteed to recover a large enough
matching. We now make this intuition formal.

We say a vertex $v \in V$ is good if it belongs to some $V^*_i$ for $i \in [k]$. 
We say a matching $M$ is \emph{trivial} if the total number of good vertices matched in $M$ is at most $N$.

\begin{claim}\label{clm:apx-factor}
	Let $M^{*}$ be a \emph{maximum matching} in $G$ and $M$ be any
        \emph{trivial matching}, then
	\[
	\frac{\card{M}}{\card{M^*}} \leq \frac{4}{\alpha}
	\]
\end{claim}
\begin{proof}
  Since $M^{*}$ is a maximum matching, it contains at least $\frac{k
    \cdot r}{2}$ edges (just using the induced matching between the good
  vertices of each player). On the other hand, since $M$ is a trivial
  matching, its size is at most the number of vertices shared by all
  players plus the number of good vertices matched in $M$, which is at
  most $2N$. Since $k = \alpha N / r$,
	\[
	\frac{\card{M}}{\card{M^*}} \leq  \frac{2N}{k \cdot r /2} 
	\leq \frac{2N}{\alpha N / 2} \leq {4 \over \alpha} 
	\]	
\end{proof}

Our goal is to prove that in any protocol that each player transmits a ``small-size'' message, the expected number of good vertices matched by the final matching is small.
In other words, the coordinator would only be able to recover a trivial matching.  

Recall that $G_i(V,E_i)$ is the graph given to the player $\PS{i}$ and
$V=[n]$. With a slight abuse of notation, we refer to the induced
subgraph of $G_i$ that is obtained by removing all isolated vertices
as the graph $G_i$ itself, since, this graph is effectively the real
input to the player $\PS{i}$. Moreover, note that picking the
permutation $\pi$ ensures that the labels of the vertices in $G_i$ are
chosen uniformly at random from $[n]$ and hence revealing no extra
information to the player $\PS{i}$.  Let $\FG_i$ be the set of all
possible graphs that $G_i$ can be. Since the edges of $G_i$ are
obtained through dropping half of the edges uniformly at random from
each induced matching of an $(r,t)$-RS graph, $\card{\FG_i} = {{r
    \choose \frac{r}{2}}}^{t}$. Moreover, in the input distribution,
$G_i$ is chosen from $\FG_i$ uniformly at random. 

For any subset $F \subseteq \FG_i$, we define the graph $G_F$ as the
\emph{intersection graph} of all graphs in $F$, i.e., an edge belongs
to the graph $G_F$ iff it belongs to every graph in $F$.

\begin{lemma}\label{lem:F-size}
  For any $i \in [k]$, any subset $F \subseteq \FG_i$, and any integer
  $\beta \geq 0$, let $I_\beta \subseteq [t]$ be the set of indices
  such that for any $j \in I_\beta$, $G_F$ contains at least
  $\frac{2^{\beta} \cdot r}{\alpha}$ edges from the $j$-th induced
  matching; if $\card{F} \geq 2^{(-\frac{r.t}{4\alpha\cdot \log{n}})}
  \card{\FG_i}$, then $\card{I_\beta} \leq
  \frac{t}{2^{\beta+2}\log{n}}$.
\end{lemma}

\begin{proof} 
	Let $\gamma = \card{I_\beta}$; we can upper bound the size of $F$ as follows: 
	\begin{alignat}{2}
	\card{F} &\leq {{r - \frac{2^\beta r}{\alpha}}\choose{\frac{r}{2}}}^{\gamma} \cdot {{r}\choose{\frac{r}{2}}}^{t - \gamma} \leq 
	\left(2^{-\frac{2^\beta r}{\alpha}}\cdot {{r}\choose{\frac{r}{2}}}\right)^{\gamma} \cdot {{r}\choose{\frac{r}{2}}}^{t - \gamma} \notag \\
	&= 2^{-\frac{2^\beta r\gamma}{\alpha}} \cdot {{r}\choose{\frac{r}{2}}}^{t} 
	= 2^{-\frac{2^\beta r \gamma}{\alpha}} \card{\FG_i} \notag
	\end{alignat}
	Therefore, $\gamma > \frac{t}{2^{\beta + 2}\log{n}}$ implies $\card{F} < 2^{(-\frac{r \cdot t}{4\alpha\cdot \log{n}})} \card{\FG_i}$; a contradiction.
\end{proof}

\begin{lemma}\label{lem:trivial}
	Suppose for each $i \in [k]$, the player $\PS{i}$ sends a message of size at most 
	\[
	s = \frac{r\cdot t}{5\alpha \cdot \log{n}}
	\] 
	bits to the coordinator; then, the expected number of good vertices that are matched in the matching computed by the coordinator is 
	at most $N/2$.   
\end{lemma}
\begin{proof}
        Fix an index $i \in [k]$ and a player $\PS{i}$. Let $X_i$ denote the
        random variable counting the number of good vertices that are
        matched by the coordinator from the graph provided to the player
        $\PS{i}$. In the following, we prove that
        \begin{alignat}{2}
        \Ex[X_i] \leq \frac{r}{2\alpha} \label{eq:exp}
        \end{alignat}
        Having this, for $X=\sum_{i \in [k]}X_i$, by linearity of expectation, we have
        $\Ex[X] \leq kr/2\alpha = N/2$, implying that the expected number of good vertices matched by the coordinator is at most $N/2$.  
        
        Suppose the coordinator knows all inputs to the players except for
        player $\PS{i}$, i.e., the graph $G_i$. Note that this is the maximum information the coordinator can obtain from other
        players. Define $\phi_i: \FG_i \mapsto \set{0,1}^{s}$ as the deterministic mapping used by the player $\PS{i}$ to map the
        input graph to a $s$-bit message and send it to the coordinator. 
        Define the function $\Gamma_i: \FG_i \mapsto
        2^{\FG_i}$ such that for any $G \in \FG_i$, $\Gamma_i(G) = \set{H \in \FG_i~|~ \phi_i(G) = \phi_i(H)}$.

	The important observation is that since the protocol is deterministic, the coordinator can output an edge $e \in G_i$ as a matching edge for the player $\PS{i}$, only if $e$ is part of \emph{every} graph 
	in $\Gamma_i(G_i)$. We define $\event_i$ to be the event that for the graph $G_i$, $\card{\Gamma_i(G_i)} < 2^{(-\frac{r.t}{4\alpha\cdot \log{n}})} \card{\FG_i}$.
	
	The following claim can be proven using a simple counting
        argument (see Appendix~\ref{app:buckets} for a proof).
        \begin{claim}\label{clm:buckets}
          For any $i \in [k]$, $\Pr(\event_i) < \frac{1}{n}$.
        \end{claim}

         We can write the expected value of $X_i$ as,
	 \begin{alignat}{2}
	\Ex[X_i] = \Ex[X_i~|~\event_i] \cdot \Pr(\event_i) + \Ex[X_i~|~\bar{\event_i}] \cdot (1-\Pr(\event_i)) \label{eq:two-term}
	\end{alignat}
	
	By Claim~\ref{clm:buckets}, the first term in this equation is less than $1$. For simplicity, we neglect this additive value of $1$ in Equation (\ref{eq:two-term}). We now bound the second term. 
	We have: 
	\begin{alignat}{2}
	\Ex[X_i~|~\bar{\event_i}] = \sum_{j = 1}^{n} j \cdot \Pr(X_i = j~|~\bar{\event_i}) \leq \sum_{\beta =1}^{\log{n}} \frac{2^{\beta+1}r}{\alpha} \cdot \Pr(X_i \geq \frac{2^{\beta}r}{\alpha}~|~\bar{\event_i}) \label{eq:pr1}
	\end{alignat}

        We can now compute $\Pr(X_i \geq  \frac{2^{\beta}r}{\alpha}~|~\bar{\event_i})$ for any $\beta
        \geq 0$ as follows. Let $F = \Gamma_i(G_i)$; the event
        $\bar{\event_i}$ implies that $\card{F} \ge   2^{(-\frac{r.t}{4\alpha\cdot \log{n}})} \card{\FG_i}$.  By Lemma~\ref{lem:F-size}, for $I_\beta$ defined as in the lemma
        statement, $\card{I_\beta} \leq \frac{t}{2^{\beta+2}\log{n}}$. In the input distribution, $\lambda$ is chosen from $[t]$ uniformly at random. Therefore,
        the probability that $\lambda \in I_\beta$ is at most $\frac{1}{2^{\beta+2}\log{n}}$.  Hence,
	\begin{alignat}{2}
	\Pr(X_i \geq \frac{2^{\beta}r}{\alpha}~|~\bar{\event_i}) = \Pr(\lambda \in I_\beta) \leq \frac{1}{2^{\beta+2}\log{n}} \label{eq:pr2}
	\end{alignat}
	By plugging in inequality (\ref{eq:pr2}) in (\ref{eq:pr1}) 
	we obtain,
	 \[
	 \Ex[X_i~|~\bar{\event_i}] \leq \sum_{\beta =1}^{\log{n}} \frac{2^{\beta+1}r}{\alpha} \cdot \frac{1}{2^{\beta+2}\log{n}} = \sum_{\beta =1}^{\log{n}} \frac{r}{2\alpha \log{n}} = \frac{r}{2\alpha}
	 \]
	Consequently, we proved the inequality (\ref{eq:exp}), i.e, $\Ex[X_i] \leq \frac{r}{2\alpha}$.  
\end{proof}

\begin{proof}(Theorem~\ref{thm:lower})
	By Lemma~\ref{lem:trivial}, if no player communicates a message of size $\Omega(\frac{r \cdot t}{\alpha \cdot \log{n}})$ bits, then the expected number of good vertices matched in the matching output 
	by the coordinator is $N/2$ and hence by Markov bound the output matching is a trivial matching with probability $1/2$. By Claim~\ref{clm:apx-factor}, any trivial matching is at most 
	an $(\alpha/4)$-approximation to the maximum matching.  
	
	Since $\alpha = n^{\eps}$, $k = n^{\eps + o(1)}$, $N = n/k$,
        and $r\cdot t = \Omega(N^2)$ (by Theorem~\ref{thm:alon-rs}),
        we have that any simultaneous protocol that obtains a better
        than $(n^{\eps}/4)$-approximation to the maximum matching with
        constant probability, has to communicate
        $n^{2-3\eps - o(1)}$ bits from at least one player.
        
\end{proof}

%% file: conclusion.tex
\section{Conclusions}\label{sec:conc}
In this paper, we resolved the space complexity of linear sketches for
approximating the maximum matching problem in dynamic graph
streams. In particular, for approximating the maximum matching to
within a factor of $n^{\eps}$, we proved that the space of $n^{2-3\eps
  \pm o(1)}$ bits is sufficient and necessary for every single-pass
streaming algorithm that only maintains a linear sketch of the
stream. 

Our result suggests that to achieve better upper bound for the maximum
matching problem, a new set of techniques is required. Alternatively,
it might be the case that any algorithm for dynamic graph streams can
be implemented as a linear sketch (similar to the equivalence between
linear sketches and single-pass turnstile
algorithms~\cite{LNW14}). As noted earlier, to the best of our knowledge,
every known single-pass streaming algorithm for the general dynamic
graph streams is indeed of this form (i.e., only maintains a linear
sketch).  In that case, our bounds would characterize the power of
\emph{any} single-pass streaming algorithm for the maximum matching
problem in dynamic graph streams.

\subsection*{Acknowledgments} 
We would like to thank Michael Kapralov and David Woodruff for helpful discussions.

%% file: appendix.tex
\section{Omitted Proofs}\label{app:omitted}

\subsection{Omitted proofs from Lemma~\ref{lem:large-spanning}}\label{app:full-ball}
\begin{claim*}
  Suppose we assign $x$ balls to $y$ bins independently and uniformly
  at random. With probability at least $1/2$, the number of non-empty
  bins is at least $\min\set{x,y}/3$.
\end{claim*}
\begin{proof}
  For each bin, the probability that the bin is empty is at most,
  \[
  (1 - {1 \over y})^x \le e^{-{x \over y}}
  \]
  
  We consider two cases. If $x/y \ge 1.5$,
  \[
   e^{-{x \over y}} \le {e^{-1.5}} < {1 \over 4}
  \]
  Hence the expected number of empty bins is at most $y/4$, and by
  Markov inequality, with probability at least $1/2$, the number of
  empty bins is at most $y/2$.

  If $x/y < 1.5$, since $e^{-z} \le 1 - z/2$ for $z \in [0, 1.5]$,
  \[
  e^{-{x \over y}} \le 1 - {x \over 2y}
  \]
  Hence the probability that a bin is non-empty is at least $x/(2y)$,
  and the expected number of non-empty bins is at least $x/2$. Since a
  bin being non-empty is negatively correlated with other bins being
  non-empty, by the extended Chernoff bound, with probability at least
  $1 - e^{-\Omega(x)}$, the number of non-empty bins is at least $x/3$.

  Hence over all, the number of non-empty bins is at least
  $\min\set{x,y}/3$ with probability at least $1/2$.
\end{proof}

\subsection{Omitted proofs from Lemma~\ref{lem:trivial}}\label{app:buckets}
\begin{claim*} For any $i \in [k]$, $\Pr(\event_i) < \frac{1}{n}$.
\end{claim*}
\begin{proof}
  Let $o \in \set{0,1}^{s}$ be the output of the function $\phi_i$,
  and with slight abuse of notation, we let $\Gamma_i(o) =
  \Gamma_i(G)$ for some $G$ such that $\phi_i(G) = o$. We say $o$ is
  \emph{light} iff $\card{\Gamma_i(o)} < 2^{(-\frac{r.t}{4\alpha\cdot
      \log{n}})} \card{\FG_i}$. We have
  \begin{alignat}{2}
    \Pr(\event_i) &= \sum_{o \in \set{0,1}^{s}\text{: $o$ is light}} \Pr_{G \sim \FG_i}(\phi_i(G)=o) \notag \\
    &= \sum_{o \in \set{0,1}^{s}\text{: $o$ is light}} \frac{\card{\Gamma_i(o)}}{\card{\FG_i}} \notag \\
    &\leq 2^{s-\frac{r.t}{4\alpha\cdot \log{n}}} < \frac{1}{n} \notag
  \end{alignat}
\end{proof}

%% file: main.bbl
\begin{thebibliography}{10}

\bibitem{TurnstileMatchingOP}
Bertinoro workshop 2014, problem 64.
\newblock \url{http://sublinear.info/index.php?title=Open_Problems:64}.
\newblock Accessed: 2015-05-1.

\bibitem{ACHPWY13}
{\sc Agarwal, P.~K., Cormode, G., Huang, Z., Phillips, J.~M., Wei, Z., and Yi,
  K.}
\newblock Mergeable summaries.
\newblock {\em {ACM} Trans. Database Syst. 38}, 4 (2013), 26.

\bibitem{AhnG13}
{\sc Ahn, K.~J., and Guha, S.}
\newblock Access to data and number of iterations: Dual primal algorithms for
  maximum matching under resource constraints.
\newblock {\em CoRR abs/1307.4359\/} (2013).

\bibitem{AG13}
{\sc Ahn, K.~J., and Guha, S.}
\newblock Linear programming in the semi-streaming model with application to
  the maximum matching problem.
\newblock {\em Inf. Comput. 222\/} (2013), 59--79.

\bibitem{AhnGM12Linear}
{\sc Ahn, K.~J., Guha, S., and McGregor, A.}
\newblock Analyzing graph structure via linear measurements.
\newblock In {\em Proceedings of the Twenty-third Annual ACM-SIAM Symposium on
  Discrete Algorithms\/} (2012), SODA '12, SIAM, pp.~459--467.

\bibitem{AGM12}
{\sc Ahn, K.~J., Guha, S., and McGregor, A.}
\newblock Graph sketches: sparsification, spanners, and subgraphs.
\newblock In {\em Proceedings of the 31st {ACM} {SIGMOD-SIGACT-SIGART}
  Symposium on Principles of Database Systems, {PODS} 2012, Scottsdale, AZ,
  USA, May 20-24, 2012\/} (2012), pp.~5--14.

\bibitem{AlonMS96}
{\sc Alon, N., Matias, Y., and Szegedy, M.}
\newblock The space complexity of approximating the frequency moments.
\newblock In {\em STOC\/} (1996), ACM, pp.~20--29.

\bibitem{AlonMS12}
{\sc Alon, N., Moitra, A., and Sudakov, B.}
\newblock Nearly complete graphs decomposable into large induced matchings and
  their applications.
\newblock In {\em Proceedings of the 44th Symposium on Theory of Computing
  Conference, {STOC} 2012, New York, NY, USA, May 19 - 22, 2012\/} (2012),
  pp.~1079--1090.

\bibitem{ANRW15}
{\sc Alon, N., Nisan, N., Raz, R., and Weinstein, O.}
\newblock Welfare maximization with limited interaction.
\newblock {\em Electronic Colloquium on Computational Complexity {(ECCC)} 22\/}
  (2015), 54.

\bibitem{ANPW13}
{\sc Andoni, A., Nguy{\^{e}}n, H.~L., Polyanskiy, Y., and Wu, Y.}
\newblock Tight lower bound for linear sketches of moments.
\newblock In {\em Automata, Languages, and Programming - 40th International
  Colloquium, {ICALP} 2013, Riga, Latvia, July 8-12, 2013, Proceedings, Part
  {I}\/} (2013), pp.~25--32.

\bibitem{BirkLM93}
{\sc Birk, Y., Linial, N., and Meshulam, R.}
\newblock On the uniform-traffic capacity of single-hop interconnections
  employing shared directional multichannels.
\newblock {\em {IEEE} Transactions on Information Theory 39}, 1 (1993),
  186--191.

\bibitem{CCHM15}
{\sc Chitnis, R.~H., Cormode, G., Hajiaghayi, M.~T., and Monemizadeh, M.}
\newblock Parameterized streaming: Maximal matching and vertex cover.
\newblock In {\em Proceedings of the Twenty-Sixth Annual {ACM-SIAM} Symposium
  on Discrete Algorithms, {SODA} 2015, San Diego, CA, USA, January 4-6, 2015\/}
  (2015), pp.~1234--1251.

\bibitem{CS14}
{\sc Crouch, M., and Stubbs, D.~S.}
\newblock Improved streaming algorithms for weighted matching, via unweighted
  matching.
\newblock In {\em Approximation, Randomization, and Combinatorial Optimization.
  Algorithms and Techniques, {APPROX/RANDOM} 2014, September 4-6, 2014,
  Barcelona, Spain\/} (2014), pp.~96--104.

\bibitem{DNO14}
{\sc Dobzinski, S., Nisan, N., and Oren, S.}
\newblock Economic efficiency requires interaction.
\newblock In {\em Symposium on Theory of Computing, {STOC} 2014, New York, NY,
  USA, May 31 - June 03, 2014\/} (2014), pp.~233--242.

\bibitem{EKS09}
{\sc Eggert, S., Kliemann, L., and Srivastav, A.}
\newblock Bipartite graph matchings in the semi-streaming model.
\newblock In {\em Algorithms - {ESA} 2009, 17th Annual European Symposium,
  Copenhagen, Denmark, September 7-9, 2009. Proceedings\/} (2009),
  pp.~492--503.

\bibitem{ELMS11}
{\sc Epstein, L., Levin, A., Mestre, J., and Segev, D.}
\newblock Improved approximation guarantees for weighted matching in the
  semi-streaming model.
\newblock {\em {SIAM} J. Discrete Math. 25}, 3 (2011), 1251--1265.

\bibitem{EHLMO15}
{\sc Esfandiari, H., Hajiaghayi, M.~T., Liaghat, V., Monemizadeh, M., and Onak,
  K.}
\newblock Streaming algorithms for estimating the matching size in planar
  graphs and beyond.
\newblock In {\em Proceedings of the Twenty-Sixth Annual {ACM-SIAM} Symposium
  on Discrete Algorithms, {SODA} 2015, San Diego, CA, USA, January 4-6, 2015\/}
  (2015), pp.~1217--1233.

\bibitem{FKMSZ05}
{\sc Feigenbaum, J., Kannan, S., McGregor, A., Suri, S., and Zhang, J.}
\newblock On graph problems in a semi-streaming model.
\newblock {\em Theor. Comput. Sci. 348}, 2-3 (2005), 207--216.

\bibitem{FischerLNRRS02}
{\sc Fischer, E., Lehman, E., Newman, I., Raskhodnikova, S., Rubinfeld, R., and
  Samorodnitsky, A.}
\newblock Monotonicity testing over general poset domains.
\newblock In {\em Proceedings on 34th Annual {ACM} Symposium on Theory of
  Computing, May 19-21, 2002, Montr{\'{e}}al, Qu{\'{e}}bec, Canada\/} (2002),
  pp.~474--483.

\bibitem{FrahlingIS2008}
{\sc Frahling, G., Indyk, P., and Sohler, C.}
\newblock Sampling in dynamic data streams and applications.
\newblock {\em International Journal of Computational Geometry \& Applications
  18}, 01n02 (2008), 3--28.

\bibitem{GoelKK12}
{\sc Goel, A., Kapralov, M., and Khanna, S.}
\newblock On the communication and streaming complexity of maximum bipartite
  matching.
\newblock In {\em Proceedings of the Twenty-third Annual ACM-SIAM Symposium on
  Discrete Algorithms\/} (2012), SODA '12, SIAM, pp.~468--485.

\bibitem{GO13}
{\sc Guruswami, V., and Onak, K.}
\newblock Superlinear lower bounds for multipass graph processing.
\newblock In {\em Proceedings of the 28th Conference on Computational
  Complexity, {CCC} 2013, K.lo Alto, California, USA, 5-7 June, 2013\/} (2013),
  pp.~287--298.

\bibitem{HW13}
{\sc Hardt, M., and Woodruff, D.~P.}
\newblock How robust are linear sketches to adaptive inputs?
\newblock In {\em Symposium on Theory of Computing Conference, STOC'13, Palo
  Alto, CA, USA, June 1-4, 2013\/} (2013), pp.~121--130.

\bibitem{HRVZ15}
{\sc Huang, Z., Radunovic, B., Vojnovic, M., and Zhang, Q.}
\newblock Communication complexity of approximate matching in distributed
  graphs.
\newblock In {\em 32nd International Symposium on Theoretical Aspects of
  Computer Science, {STACS} 2015, March 4-7, 2015, Garching, Germany\/} (2015),
  pp.~460--473.

\bibitem{JowhariST2011}
{\sc Jowhari, H., Sa{\u{g}}lam, M., and Tardos, G.}
\newblock Tight bounds for lp samplers, finding duplicates in streams, and
  related problems.
\newblock In {\em Proceedings of the thirtieth ACM SIGMOD-SIGACT-SIGART
  symposium on Principles of database systems\/} (2011), ACM, pp.~49--58.

\bibitem{K13}
{\sc Kapralov, M.}
\newblock Better bounds for matchings in the streaming model.
\newblock In {\em Proceedings of the Twenty-Fourth Annual {ACM-SIAM} Symposium
  on Discrete Algorithms, {SODA} 2013, New Orleans, Louisiana, USA, January
  6-8, 2013\/} (2013), pp.~1679--1697.

\bibitem{KKS14}
{\sc Kapralov, M., Khanna, S., and Sudan, M.}
\newblock Approximating matching size from random streams.
\newblock In {\em Proceedings of the Twenty-Fifth Annual {ACM-SIAM} Symposium
  on Discrete Algorithms, {SODA} 2014, Portland, Oregon, USA, January 5-7,
  2014\/} (2014), pp.~734--751.

\bibitem{KLMMS14}
{\sc Kapralov, M., Lee, Y.~T., Musco, C., Musco, C., and Sidford, A.}
\newblock Single pass spectral sparsification in dynamic streams.
\newblock In {\em 55th {IEEE} Annual Symposium on Foundations of Computer
  Science, {FOCS} 2014, Philadelphia, PA, USA, October 18-21, 2014\/} (2014),
  pp.~561--570.

\bibitem{Konrad15}
{\sc Konrad, C.}
\newblock Maximum matching in turnstile streams.
\newblock {\em Manuscript, \textnormal{May, 2015}\/}.

\bibitem{KMM12}
{\sc Konrad, C., Magniez, F., and Mathieu, C.}
\newblock Maximum matching in semi-streaming with few passes.
\newblock In {\em Approximation, Randomization, and Combinatorial Optimization.
  Algorithms and Techniques - 15th International Workshop, {APPROX} 2012, and
  16th International Workshop, {RANDOM} 2012, Cambridge, MA, USA, August 15-17,
  2012. Proceedings\/} (2012), pp.~231--242.

\bibitem{CCbook}
{\sc Kushilevitz, E., and Nisan, N.}
\newblock {\em Communication complexity}.
\newblock Cambridge University Press, 1997.

\bibitem{LNW14}
{\sc Li, Y., Nguyen, H.~L., and Woodruff, D.~P.}
\newblock Turnstile streaming algorithms might as well be linear sketches.
\newblock In {\em Symposium on Theory of Computing, {STOC} 2014, New York, NY,
  USA, May 31 - June 03, 2014\/} (2014), pp.~174--183.

\bibitem{M05}
{\sc McGregor, A.}
\newblock Finding graph matchings in data streams.
\newblock In {\em Approximation, Randomization and Combinatorial Optimization,
  Algorithms and Techniques, 8th International Workshop on Approximation
  Algorithms for Combinatorial Optimization Problems, {APPROX} 2005 and 9th
  InternationalWorkshop on Randomization and Computation, {RANDOM} 2005,
  Berkeley, CA, USA, August 22-24, 2005, Proceedings\/} (2005), pp.~170--181.

\bibitem{M14}
{\sc McGregor, A.}
\newblock Graph stream algorithms: a survey.
\newblock {\em {SIGMOD} Record 43}, 1 (2014), 9--20.

\bibitem{RAbook}
{\sc Motwani, R., and Raghavan, P.}
\newblock {\em Randomized Algorithms}.
\newblock Cambridge University Press, 1995.

\bibitem{Muth05}
{\sc Muthukrishnan, S.}
\newblock Data streams: Algorithms and applications.
\newblock {\em Foundations and Trends in Theoretical Computer Science 1}, 2
  (2005).

\bibitem{RuszaS78}
{\sc Ruzsa, I.~Z., and Szemer{\'e}di, E.}
\newblock Triple systems with no six points carrying three triangles.
\newblock {\em Combinatorics (Keszthely, 1976), Coll. Math. Soc. J. Bolyai
  18\/} (1978), 939--945.

\bibitem{SchmidtSS95}
{\sc Schmidt, J.~P., Siegel, A., and Srinivasan, A.}
\newblock Chernoff-hoeffding bounds for applications with limited independence.
\newblock {\em {SIAM} J. Discrete Math. 8}, 2 (1995), 223--250.

\bibitem{W14}
{\sc Woodruff, D.~P.}
\newblock Sketching as a tool for numerical linear algebra.
\newblock {\em Foundations and Trends in Theoretical Computer Science 10}, 1-2
  (2014), 1--157.

\bibitem{Z12}
{\sc Zelke, M.}
\newblock Weighted matching in the semi-streaming model.
\newblock {\em Algorithmica 62}, 1-2 (2012), 1--20.

\end{thebibliography}
